\documentclass[envcountsame]{llncs}

\usepackage[utf8]{inputenc}
\usepackage{enumerate}
\usepackage{graphicx}
\usepackage{amsfonts}
\usepackage{amssymb}
\usepackage{amsmath}
\usepackage{wasysym}
\usepackage{hyperref}
\usepackage{mathtools}
\usepackage{bibentry}
\usepackage{comment}
\usepackage{enumitem}
\usepackage{todonotes}
\usepackage{hyperref}
\usepackage{cite}
\usepackage[numbers,sort]{natbib}
\usepackage{cleveref}

\def\lhom#1{\textsc{List $#1$-Colour\-ing}}

\usepackage{amsthm}

\title{List homomorphisms to separable signed graphs}

\author{Jan Bok\inst{1,6}\orcidID{0000-0002-7973-1361} \and Richard Brewster\inst{2}\orcidID{0000-0001-7237-4288} \and Tom\' as Feder\inst{3} \and Pavol Hell\inst{4}\orcidID{0000-0001-7609-9746} \and Nikola Jedličková\inst{5}\orcidID{0000-0001-9518-6386}}

\institute{
Computer Science Institute, Faculty of Mathematics and Physics, Charles University, Prague, Czech Republic, \url{bok@iuuk.mff.cuni.cz}
\and
Department of Mathematics and Statistics, Thompson Rivers University, Canada, \url{rbrewster@tru.ca}
\and
268 Waverley St., Palo Alto, USA, \url{tomas@theory.stanford.edu}
\and
School of Computing Science, Simon Fraser University, Canada, \url{pavol@sfu.ca}
\and
Department of Applied Mathematics, Faculty of Mathematics and Physics, Charles University, Prague, Czech Republic, \url{jedlickova@kam.mff.cuni.cz}
\and
Universit\'e Clermont Auvergne, CNRS, Clermont Auvergne INP, Mines Saint-\'Etienne, LIMOS, 63000 Clermont-Ferrand, France
}

\begin{document}

\maketitle
\begin{abstract}
The complexity of the list homomorphism problem for signed graphs appears difficult to classify. Existing results focus on special classes of signed graphs, such as trees \cite{mfcs} and reflexive signed graphs \cite{ks}. Irreflexive signed graphs are in a certain sense the heart of the problem, as noted by a recent paper of Kim and Siggers. We focus on a special class of irreflexive signed graphs, namely those in which the unicoloured edges form a spanning path or cycle, which we call separable signed graphs. We classify the complexity of list homomorphisms to these separable signed graphs; we believe that these signed graphs will play an important role for the general resolution of the irreflexive case. We also relate our results to a conjecture of Kim and Siggers concerning the special case of semi-balanced irreflexive signed graphs; we have proved the conjecture in another paper, and the present results add structural information to that topic.
\end{abstract}

\section{Motivation and background}

We investigate the complexity of (list) homomorphism problems for signed graphs. The complexity of homomorphism (and list homomorphism) problems is a popular topic. For undirected graphs, it was shown in~\cite{hn} that the problem of deciding the existence of a homomorphism of an input graph to a fixed graph $H$ (also known as the \textsc{$H$-Colouring} problem, or just \textsc{$H$-Colouring}) is polynomial if $H$ is bipartite or has a loop, and is NP-complete otherwise. For general structures $H$, the corresponding problem lead to the so-called Dichotomy Conjecture~\cite{fv,jeav}, which was only recently established~\cite{bula,zhuk}. In the list homomorphism problem for $H$ (also known as the \lhom{H} problem, or just \lhom{H}), the input is a graph together with lists of allowed images for each vertex. (The precise definitions are given below.) The list homomorphism problems have generally a nicer behaviour than the homomorphism problems, because the lists facilitate recursion to subproblems. For undirected graphs, \lhom{H} is polynomial if $H$ is a bi-arc graph (see below), and is NP-complete otherwise~\cite{feder1998list,feder1999list}. Even for general structures $H$, where the list version is equivalent to a special case of the basic version, the classification for the list version was achieved a decade earlier~\cite{bulatov} than the basic version.

Signed graphs are related to graphs with two symmetric binary relations; they are additionally equipped with an operation of {\em switching} (explained below). The possibility of switching poses a challenge when classifying the complexity of homomorphisms, as the problem no longer appears to be a homomorphism problem for relational structures. Nevertheless, it can be shown~\cite{mfcs} that it is equivalent to such a problem and hence the results from~\cite{bula,zhuk} imply that there these problems also enjoy a dichotomy of polynomial versus NP-complete. For homomorphisms of signed graphs without lists, a concrete dichotomy classification was conjectured in~\cite{BFHN}, and proved in~\cite{dichotomy}. Interestingly, for signed graphs, the list version no longer seems easier to classify, and the progress towards a classification or even a conjecture has been slow~\cite{dm,BFHN,ks}.

\paragraph{Signed graphs.} A {\em signed graph} $\widehat{G}$ consists of a set $V(G)$ and two symmetric binary relations $+, -$. We also view $\widehat{G}$ as a graph $G$ (the {\em underlying graph of $\widehat{G}$}) with the vertex set $V(G)$, the edge set $+ \cup -$, and a mapping $\sigma: E(G) \to \{+, -\}$, assigning a sign ($+$ or $-$) to each edge of $G$. A loop is considered to be an edge. Two signed graphs are considered {\em (switching) equivalent} if one can be obtained from the other by a sequence of {\em switchings}; switching at a vertex $v$ results in changing the signs of all non-loop edges incident to $v$. The signs of loops are unchanged by switching.

We will usually view signs of edges as colours, and view positive edges as  {\em blue} (solid lines in figures), and negative edges as {\em red} (dashed lines in figures). It will be convenient to call a red-blue pair of edges with the same endpoint(s) a {\em bicoloured edge}; however, it is important to keep in mind that formally they are two distinct edges. By contrast, we call edges that are not part of such a pair 
{\em unicoloured}. If a vertex $u$ is adjacent by a bicoloured edge or a unicoloured edge to $v$, we say that $v$ is a \emph{bicoloured neighbour} or an \emph{unicoloured neighbour} of $u$, respectively.

We call $\widehat{H}$ a {\em signed tree} if the underlying graph $H$, with any existing loops removed and multi-edges replaced by simple edges, is a tree.

The study of signed graphs seems to have originated in~\cite{harary,hararykabell}, and was most notably advanced in the papers of Zaslavsky~\cite{zav81,zav82b,zav82a,Z97,zavsurvey}. Guenin~\cite{guenin} pioneered the investigation of homomorphisms of signed graphs; see also, e.g.,~\cite{brewgrav,nasrolsop,rezazasla}. 

\paragraph{Homomorphisms of signed graphs.} A {\em sign-preserving homomorphism} of a signed graph $\widehat{G}$ to a signed graph $\widehat{H}$ is a function $f: V(G) \to V(H)$ such that if $xy$ is a blue (respecitvely red) edge of $\widehat{G}$, then $f(x)f(y)$ is either a blue (respectively red) or a bicoloured edge of $\widehat{H}$.  This definition implies bicoloured edges of $\widehat{G}$ map to bicoloured edges of $\widehat{H}$ and for each unicoloured edge of $\widehat{G}$ with a bicoloured image in $\widehat{H}$ implicitly the image is the appropriate edge of the same sign within the bicoloured edge.
A {\em homomorphism} of the signed graph $\widehat{G}$ to the signed graph $\widehat{H}$ is a mapping $f : V(G) \to V(H)$ for which there exists a signed graph $\widehat{G}'$ equivalent to $\widehat{G}$ such that $f$ is a sign-preserving homomorphism of $\widehat{G}'$ to $\widehat{H}$. A {\em list homomorphism} of $\widehat{G}$ to $\widehat{H}$, with respect to the lists $L(v) \subseteq V(H), v \in V(G)$, is a homomorphism $f$ of $\widehat{G}$ to $\widehat{H}$ such that $f(v) \in L(v)$ for all $v \in V(G)$.  We remark an equivalent definition of homomorphisms for signed graphs without switching but rather based on the balance of cycles (in $\widehat{G}$ and their image in $\widehat{H}$) can be found in~\cite{rezazasla}. This alternative definition works with the notions of unicoloured and bicoloured edges. See~\cite{arxiv-latin} for further details.  For this paper we will use exclusively the definition based on switching and sign-preserving homomorphisms.

Let $\widehat{H}$ be a fixed signed graph.
The {\em homomorphism problem} for $\widehat{H}$ (the \textsc{$\widehat{H}$-Colouring} problem, or just \textsc{$H$-Colouring}) takes as input a signed graph $\widehat{G}$ and asks whether there exists a homomorphism of $\widehat{G}$ to $\widehat{H}$. The {\em list homomorphism problem} for $\widehat{H}$ (the \lhom{\widehat{H}} problem, or just \lhom{\widehat{H}}) takes as an input a signed graph $\widehat{G}$ with lists $L(v) \subseteq V(H)$, for every $v \in V(G)$, and asks whether there exists a homomorphism $f$ of the signed graph $\widehat{G}$ to $\widehat{H}$ such that $f(v) \in L(v)$ for every $v \in V(G)$. 

A subgraph $\widehat{G}$ of the signed graph $\widehat{H}$ is the \emph{signed core}, or simply an {\em s-core}, of $\widehat{H}$ if there is signed graph homomorphism  $f$ of $\widehat{H}$ to $\widehat{G}$, and every homomorphism of the signed graph $\widehat{G}$ to itself is a bijection on $V(G)$. It is easy to see that the signed core of any signed graph is unique up to switching isomorphism\footnote{$\widehat G$ and $\widehat H$ are switching isomorphic if there exist homomorphisms $\phi: \widehat G \to \widehat H$ and $\psi: \widehat H \to \widehat G$ such that $\phi \circ \psi$ and $\psi \circ \phi$ are identity mappings on $\widehat G$ and $\widehat H$, respectively.}
The dichotomy classification for \textsc{$H$-Colouring} conjectured in~\cite{BFHN} and proved in~\cite{dichotomy} is as follows. 

\begin{theorem}{\cite{dichotomy}}\label{thm:dichotomy}
\textsc{$H$-Colouring} is po\-ly\-no\-mial-time solvable if the signed core of $\widehat{H}$ has at most two edges, and is NP-complete otherwise.
\end{theorem}

In counting edges we of course include loops, and count each unicoloured edge as one and each bicoloured edge as two.

\paragraph{Balanced and semi-balanced signed graphs.} A signed graph is {\em balanced} if it is equivalent to one in which all edges are only blue, and is {\em anti-balanced} if it is equivalent to one in which all edges are only red. Note that it follows that a balanced (anti-balanced) signed graph has no bicoloured edges.

A signed graph is {\em  semi-balanced} ({\em semi-anti-balanced}) if it is equivalent to one in which all edges are bicoloured or blue (respectively red).\footnote{We note that this class has been called {\em pr-graphs} in~\cite{ks}, {\em uni-balanced graphs} in~\cite{mfcs}, and {\em weakly balanced graphs} in~\cite{dm,latin}.}

For list homomorphisms of signed graphs, there are several special cases where the complexity has been classified. These include signed graphs without bicoloured edges~\cite{bordeaux}, signed trees with possible loops~\cite{mfcs}, and semi-balanced reflexive and irreflexive signed graphs~\cite{arxiv-latin,latin,ks}. We first introduce the relevant structures used to prove NP-completeness results and used to construct polynomial time algorithms.

\paragraph{Chains.} Let $U, D$ be two walks in $\widehat{H}$ of equal length. Suppose $U$ has vertices $u = u_0,u_1, \ldots,u_k=v$, and $D$ has vertices $u = d_0,d_1, \ldots,d_k=v$. As $U$ and $D$ are walks, vertices may repeat both within $U$ and $D$ and be common to both $U$ and $D$.
We say that $(U,D)$ is a \emph{chain}, provided $uu_1, d_{k-1}v$ are unicoloured edges and $ud_1, u_{k-1}v$ are bicoloured edges, and for each~$i$, $1 \leq i \leq k-2$, we have
\begin{enumerate}
\item both $u_iu_{i+1}$ and $d_id_{i+1}$ are edges of $\widehat{H}$ 
while $d_iu_{i+1}$ is not an edge of $\widehat{H}$, or
\item both $u_iu_{i+1}$ and $d_id_{i+1}$ are bicoloured edges of $\widehat{H}$ 
while $d_iu_{i+1}$ is not a bicoloured edge of~$\widehat{H}$.
\end{enumerate}

\Cref{fig:chain} show a simple example of a chain and \Cref{fig:forbgraphsirref} shows some important irreflexive signed trees with a chain.

The existence of a chain in a signed graph implies hardness.

\begin{figure}[tb]
\centering
\includegraphics[scale=1]{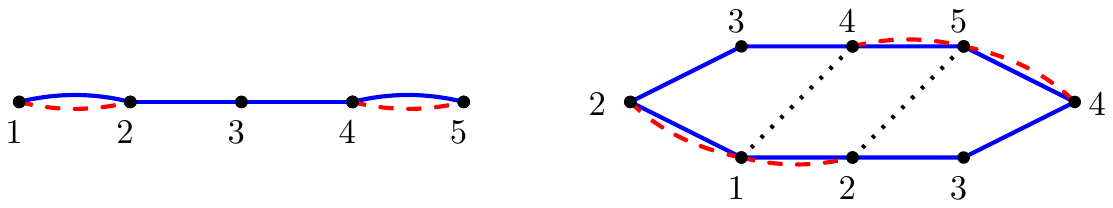}
\caption{A signed graph (on the left) together with a chain (on the right). The upper walk is $U$, the lower walk is $D$; the dotted blue edges must be absent.}
\label{fig:chain}
\end{figure}

\begin{figure}[tb]
\centering
\includegraphics[width=\textwidth]{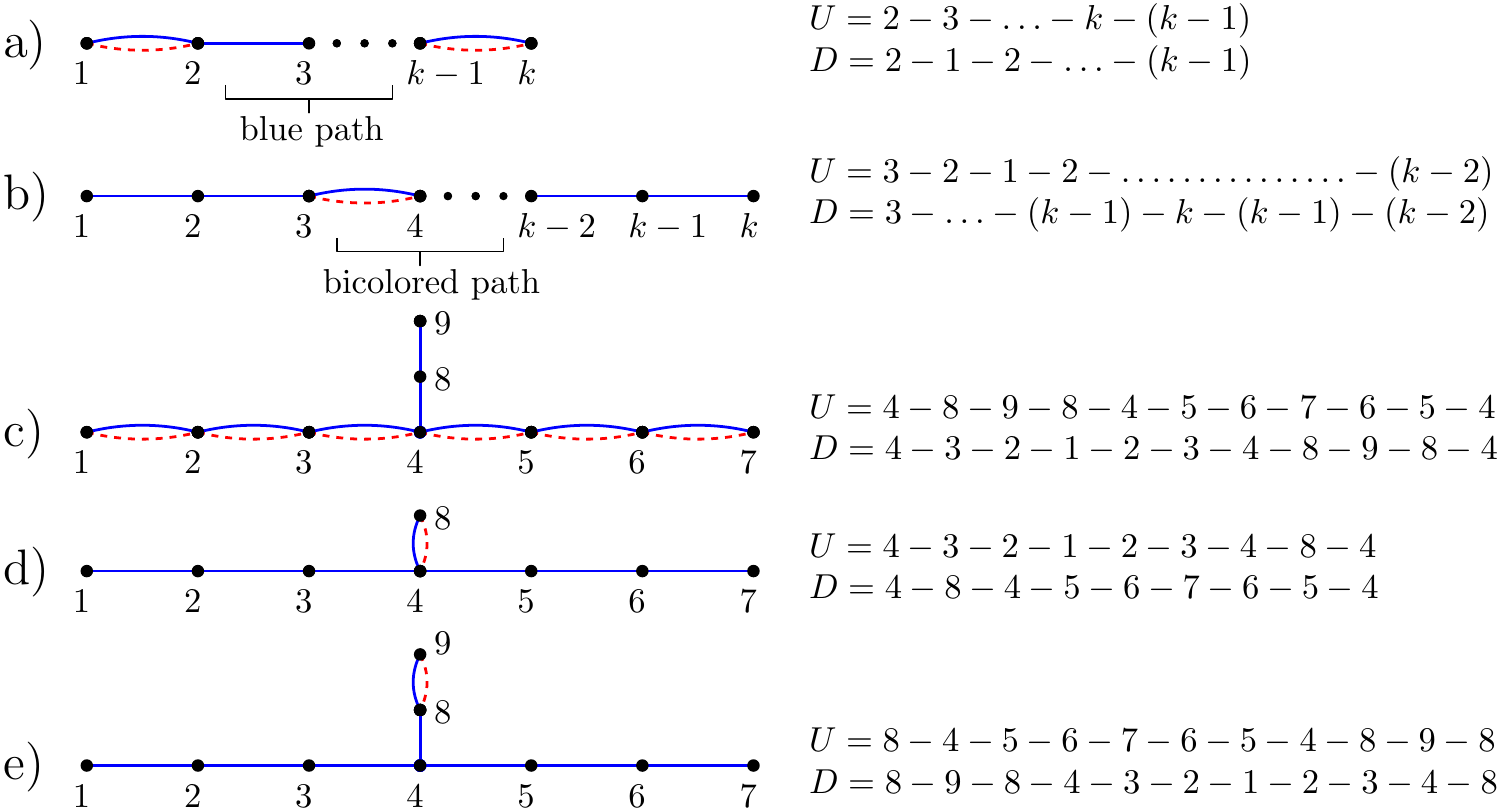}
\caption{The family $\cal F$ of signed graphs yielding NP-complete problems, and a chain in each. (The figure appeared first in~\cite{mfcs}.)}
\label{fig:forbgraphsirref}
\end{figure}

\begin{theorem}{\cite{mfcs}}\label{thm:chain}
If a signed graph $\widehat{H}$ contains a chain, then \lhom{\widehat H} is NP-complete.
\end{theorem}

\paragraph{Invertible pairs.} An {\em invertible pair} in an undirected graph $H$ is a pair of vertices $a, b$, with two walks $U, D$ of the same length, where $U$ has vertices $a = u_0,u_1,\ldots,u_k=b, u_{k+1},\ldots,u_t = a$, and $D$ has vertices $b=d_0,d_1, \ldots,d_k=a,d_{k+1},\ldots,d_t=b$, such that for each~$i$, $1 \leq i \leq t-2$, both $u_iu_{i+1}$ and $d_id_{i+1}$ are edges of $H$, while $d_iu_{i+1}$ is not an edge of $\widehat{H}$. For simplicity we say that a signed graph has an invertible pair if its underlying graph has an invertible pair.

\Cref{fig:tripleclaw} shows the graph $F_1$, with an invertible pair $1, 10$. The walks $U, D$ begin as indicated. Then $U$ alternates on $7-6$ while $D$ moves from $10$ to $1$.  Next while $D$ alternates on $1-2$, $U$ moves from $7$ to $10$.  Continuing similarly for the second half, $U$ moves from $10$ to $1$ and $D$ moves from $1$ to $10$.

\begin{figure}
\centering
\includegraphics[width=0.75\textwidth]{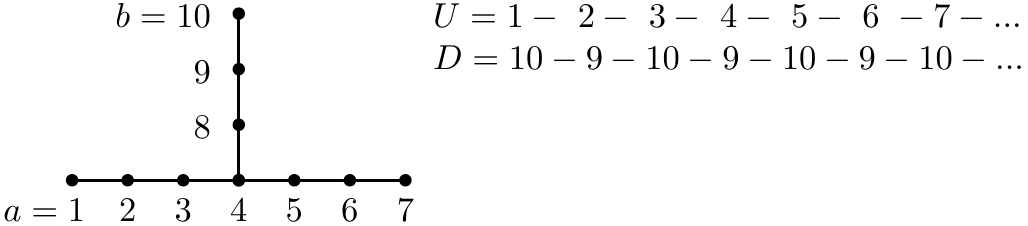}
\caption{The graph $F_1$, with an invertible pair.}
\label{fig:tripleclaw}
\end{figure}

The following result follows from ~\cite{mfcs,feder1998list,esa,feder1999list}.

\begin{theorem}\label{thm:invo}
If $\widehat{H}$ has an invertible pair, then \lhom{\widehat{H}} is NP-complete.
\end{theorem}

We now introduce the concepts needed to state the polynomial time results.

\paragraph{Bi-arc graphs.} Let $C$ be a fixed circle with two specified points $n$ and $s$. A {\em bi-arc graph} is a graph $H$ such that each vertex $v \in V(H)$ can be associated with a pair of intervals $N_v, S_v$ where $N_v$ contains $n$ but not $s$ and $S_v$ contains $s$ but not $n$ satisfying the following conditions: (i) $N_v$ intersects $S_w$ if and only if $S_v$ intersects $N_w$, and (ii) $N_v$ intersects $S_w$ if and only if $vw$ is not an edge of $H$. This class of graphs includes all interval graphs: a reflexive graph is a bi-arc graph if and only if it is an interval graph. Moreover, an irreflexive graph is a bi-arc graph if and only if it is bipartite and its complement is a circular arc graph~\cite{feder1999list}.

\paragraph{Min ordering.} To distinguish the two parts of a bipartite graph we speak of \emph{black} and \emph{white} vertices. A {\em min ordering} of a bipartite graph $H$ is a pair $<_b, <_w$, where $<_b$ is a linear ordering of the black vertices and $<_w$ is a linear ordering of the white vertices, such that for white vertices $x <_w x'$ and black vertices $y <_b y'$, if $xy', x'y$ are both edges in $H$, then $xy$ is also an edge in $H$. It is known~\cite{fv} that if a bipartite graph $H$ has a min ordering, then \lhom{H} can be solved in polynomial time. A {\em special min ordering} of a signed irreflexive graph $\widehat{H}$ is a min ordering of the underlying undirected graph $H$ for which all bicoloured neighbours of each vertex appear before all of its unicoloured neighbours.

\paragraph{Existing classifications.} We can now state our classifications mentioned above. For the case of signed graphs without bicoloured edges, the result is this.

\begin{theorem} {\cite{bordeaux}}\label{thm:bordeaux}
Suppose $\widehat{H}$ is a connected signed graph without bicoloured edges. If the underlying graph $H$ is a bi-arc graph, and $\widehat{H}$ is balanced or anti-balanced, then \lhom{\widehat{H}} is polynomial-time solvable. Otherwise, the problem is NP-complete.
\end{theorem}

For signed trees with possible loops, the general results are technical \cite{mfcs}, but since the focus of this paper is on irreflexive graphs we state the classification in the special case of irreflexive signed trees. Note $F_1$ and $\mathcal{F}$ refer to \Cref{fig:tripleclaw} and $\Cref{fig:forbgraphsirref}$ respectively.

\begin{theorem}{\cite{mfcs}} \label{thm:main_irref}
Let $\widehat{H}$ be an irreflexive tree. If the underlying graph of $\widehat{H}$ contains $F_1$ or if $\widehat{H}$ contains a signed graph from the family $\cal F$, as an induced subgraph, then \lhom{\widehat{H}} is NP-complete. Otherwise, $\widehat{H}$ admits a special min ordering and the problem is polynomial-time solvable.
\end{theorem}

Recall that the tree $F_1$ from \Cref{fig:tripleclaw} has an invertible pair. Thus for irreflexive trees, the only NP-complete cases have a chain or an invertible pair. It is easy to see that irreflexive trees are always semi-balanced. We can similarly interpret \Cref{thm:bordeaux} specialized for irreflexive signed graphs without bicoloured edges. Specifically, if the underlying graph is not bipartite then \Cref{thm:dichotomy} implies the list homomorphism problem is NP-complete. Also note that bipartite signed graphs without bicoloured edges are semi-balanced if and only if they are balanced. Finally, a bipartite graph is a bi-arc graph if and only if it has a min ordering \cite{esa}, which is trivially special as there are no bicoloured edges. Thus \Cref{thm:bordeaux} also implies that for (semi-) balanced irreflexive signed graphs without bicoloured edges, the list homomorphism problem is polynomial-time solvable if there is special min ordering \cite{feder1999list}, and otherwise there is an invertible pair and the problem is NP-complete. Note that a graph without bicoloured edges cannot have a chain.

Extrapolating from these results, Kim and Siggers \cite{ks} conjectured that for all semi-balanced irreflexive signed graphs, the list homomorphism problem is polynomial-time solvable if there is special min ordering, and otherwise there is a chain or an invertible pair and the problem is NP-complete. We have proved this conjecture 
in~\cite{latin}.

\begin{theorem}{\cite{latin}} \label{mak}
For a semi-balanced irreflexive signed graph $\widehat{H}$, \lhom{\widehat H} is polynomial-time solvable if $\widehat{H}$ has a special min ordering; otherwise, $\widehat{H}$ contains a chain or an invertible pair and the problem is NP-complete.
\end{theorem}

Similar facts have been proved for reflexive graphs \cite{arxiv-latin,ks}.

\subsection{Our results}

In this paper we focus on another particular class of irreflexive signed graphs called {\em separable} signed graphs. We say that an irreflexive signed graph $\widehat{H}$ is {\em path-separable (respectively cycle-separable)} if the unicoloured edges of $\widehat{H}$ form a spanning path (respectively cycle) in the underlying graph of $\widehat{H}$. We also say a signed graph is {\em separable} if it is path-separable or cycle-separable. 

We provide a detailed classification of complexity of the corresponding list homomorphism problems, see \Cref{hlavna,thm:Hamilton}. We find that the polynomial cases are nicely structured and rather rare. 

Separable signed graphs are not necessarily semi-balanced, so our results extend beyond \Cref{mak}. Even for semi-balanced separable graphs these results provide much more structural detail for the description of the polynomial cases. We believe that the case of separable signed graphs, together with the case of irreflexive signed trees from \Cref{thm:main_irref}, will play an important role for a full classification of complexity for irreflexive signed graphs.

A short preliminary version of this paper (containing only a small subset of the proofs) has appeared in the conference \cite{separable}, prior to our paper \cite{latin}. In this expanded journal version we provide all proofs, taking advantage of now having \Cref{mak} to simplify some of the arguments. The original proofs can be found in \cite{arXivVersion}.

\subsection{Motivation}

Irreflexive signed graphs are in a sense the core of the problem, see \cite{ks}. As noted earlier, by \Cref{thm:dichotomy}, \lhom{H} is NP-complete unless the underlying graph $H$ is bipartite. There is a natural transformation of each general problem to a problem for a bipartite irreflexive signed graph, akin to what is done for unsigned graphs in~\cite{feder2003bi}; this is nicely explained in~\cite{ks}.

However, for bipartite $H$, we don't have a combinatorial classification beyond the case of trees $H$, except in the case $\widehat{H}$ has no bicoloured edges (when \Cref{thm:bordeaux} applies), or when $\widehat{H}$ has no unicoloured edges (when the problem essentially concerns unsigned graphs and thus is solved by~\cite{feder2003bi}). Therefore we may assume that both bicoloured and unicoloured edges are present. We focus in this paper on those bipartite irreflexive signed graphs $\widehat{H}$ in which the unicoloured edges form simple structures, namely spanning paths and cycles. 

\section{Path-separable signed graphs}

In this section, we consider irreflexive path-separable signed graphs $\widehat{H}$, i.e., ones in which the unicoloured edges form a spanning path $P$ in the underlying graph $H$.
By suitable switching, we may assume the edges of $P$ are all blue. In other words, all the edges of the spanning path $P$ are blue, and all the other edges of $\widehat{H}$ are bicoloured. Thus a path-separable signed graph is semi-balanced. Recall that the distinction between unicoloured and bicoloured edges is independent of switching, thus such a spanning path $P$ is unique. 

\paragraph{Induced cycles imply hardness.} Recall that for any irreflexive signed graph $\widehat{H}$, \lhom{\widehat H} is NP-complete if the underlying graph $H$ contains an odd cycle, since then the s-core of $\widehat{H}$ has at least three edges. Moreover, we now show that the \lhom{\widehat H} is also NP-complete if $H$ contains any induced cycle of length greater than four. Indeed, it suffices to prove this if $H$ is an even cycle of length $k > 4$. If all edges of $H$ are unicoloured, then the problem is NP-complete by \Cref{thm:bordeaux}, since an irreflexive cycle of length $k > 4$ is not a bi-arc graph. If all edges of the cycle $H$ are bicoloured, then we can easily reduce from the previous case. If $H$ contains both unicoloured and bicoloured edges, then $\widehat{H}$ contains an induced subgraph of type a) or b) in the family $\cal F$ in \Cref{fig:forbgraphsirref}, and the problem is NP-complete by \Cref{thm:chain}. There are cases when the subgraphs are not induced, but the chains from the proof of \Cref{thm:chain} are still applicable.

\paragraph{Two important patterns.} We further identify two additional cases of $\widehat{H}$ with NP-complete list homomorphism problems. An \emph{alternating 4-cycle} is a 4-cycle $v_1v_2v_3v_4$ in which the edges $v_1v_2, v_3v_4$ are bicoloured and the edges $v_2v_3, v_4v_1$ unicoloured. A \emph{4-cycle pair} consists of 4-cycles $v_1v_2v_3v_4$ and $v_1v_5v_6v_7$, sharing the vertex $v_1$, in which the edges $v_1v_2, v_1v_5$ are bicoloured, and all other edges are unicoloured. An alternating 4-cycle has the chain $U = v_1, v_4, v_3; D = v_1, v_2, v_3$, and a 4-cycle pair has the chain $U = v_1, v_4, v_3, v_2, v_1; D = v_1, v_5, v_6, v_7, v_1$. Therefore, if a signed graph $\widehat{H}$ contains an alternating 4-cycle or a 4-cycle pair as an induced subgraph, then \lhom{\widehat H} is NP-complete. Note that the latter chain requires only $v_2v_6$ and $v_3v_5$ to be non-edges. The problem remains NP-complete as long as these edges are absent; all other edges with endpoints in different 4-cycles can be present. If both $v_2v_6$ and $v_3v_5$ are bicoloured edges, then there is an alternating 4-cycle $v_2v_3v_5v_6$. Thus we conclude that \lhom{\widehat H} is NP-complete if $\widehat{H}$ contains a 4-cycle pair as a subgraph (not necessarily induced), unless exactly one of $v_3v_5$ or $v_2v_6$ is a bicoloured edge.

\paragraph{Assumptions.} From now on we will assume that $\widehat{H}$ is a path-separable signed graph with the unicoloured edges (all blue) forming a spanning path $P = v_1, \ldots, v_n$. We will assume further that \lhom{\widehat H} is not NP-complete, and derive information on the structure of $\widehat{H}$. In particular, the underlying graph $H$ is bipartite and does not contain any induced cycles of length greater than $4$, and $\widehat{H}$ does not contain an alternating $4$-cycle or a $4$-cycle pair; more generally, $\widehat{H}$ does not contain a chain. If $\widehat{H}$ has no bicoloured edges (and hence no edges not on $P$), then \lhom{\widehat H} is polynomial-time solvable by \Cref{thm:bordeaux}, since a path is a bi-arc graph. If there is a bicoloured edge in $\widehat{H}$, then we may assume there is a bicoloured edge $v_iv_{i+3}$, otherwise there is an induced cycle of length greater than $4$.

\paragraph{Blocks and segments.} A {\em block} in a path-separable signed graph $\widehat{H}$ is a subpath $v_iv_{i+1}v_{i+2}v_{i+3}$ of $P$, with the bicoloured edge $v_iv_{i+3}$. The previous paragraph concluded that $\widehat{H}$ must contain a block. Note that if $v_iv_{i+1}v_{i+2}v_{i+3}$ is a block, then $v_{i+1}v_{i+2}v_{i+3}v_{i+4}$ cannot be a block: in fact, $v_{i+1}v_{i+4}$ cannot be a bicoloured edge, otherwise $\widehat{H}$ would contain an alternating $4$-cycle. However, $v_{i+2}v_{i+3}v_{i+4}v_{i+5}$ can again be a block, and so can $v_{i+4}v_{i+5}v_{i+6}v_{i+7}$, etc. If both $v_iv_{i+1}v_{i+2}v_{i+3}$ and $v_{i+2}v_{i+3}v_{i+4}v_{i+5}$ are blocks then $v_iv_{i+5}$ must be a bicoloured edge, otherwise $v_iv_{i+3}v_{i+2}v_{i+5}$ would induce a signed graph of type a) in family $\cal F$ from \Cref{fig:forbgraphsirref}. A {\em segment} in $\widehat{H}$ is a maximal subpath $v_iv_{i+1} \ldots v_{i+2j+1}$ of $P$ with $j \geq 1$ that has all bicoloured edges $v_{i+e}v_{i+e+3}$, where $e$ is even, $0 \leq e \leq 2j-2$. A {\em maximal} subpath is not properly contained in another such subpath. Thus each subpath $v_{i+e}v_{i+e+1}v_{i+e+2}v_{i+e+3}$ of the segment is a block, and the segment is a consecutively intersecting sequence of blocks; note that it can consist of just one block. Two segments can touch as the second and third segment in \Cref{fig:pathsegment}, or leave a gap as the first and second segment in the same figure.

\paragraph{Sources.} In a segment $v_iv_{i+1} \ldots v_{i+2j+1}$ we call each vertex $v_{i+e}$ with $0 \leq e \leq 2j-2$ a {\em forward source}, and each vertex $v_{i+o}$ with $3 \leq o \leq 2j+1$ a {\em backward source}. Thus forward sources are the beginning vertices of blocks in the segment, and the backward sources are the ends of blocks in the segment. If $a < b$, we say the edge $v_av_b$ is a {\em forward edge} from $v_a$ and a {\em backward edge} from $v_b$. In this terminology, each forward source has a forward edge to its corresponding backward source. Because of the absence of a signed graph of type a) in family $\cal F$ from \Cref{fig:forbgraphsirref}, we can in fact conclude, by the same argument as in the previous paragraph, that each forward source in a segment has forward edges to all backward sources in the segment.

\paragraph{Leaning segments.} We say that a segment $v_iv_{i+1} \ldots v_{i+2j+1}$ is {\em right-leaning} if $v_{i+e}v_{i+e+o}$ is a bicoloured edge for all $e$ is even, $0 \leq e \leq 2j-2$, and {\em all} odd $o \geq 3$; and we say it is {\em left-leaning} if $v_{i+2j+1-e}v_{i+2j+1-e-o}$ is a bicoloured edge for all $e$ even, $0 \leq e \leq 2j-2$ and all odd $o \geq 3$. Thus in a right-leaning segment each forward source has {\em all} possible forward edges (that is, all edges to vertices of opposite colour in the bipartition, including vertices with subscripts greater than $i+2j+1$). The concepts of left-leaning segments, backward sources and backward edges are defined similarly.

\paragraph{Segmented graphs.} We say that a path-separable signed graph $\widehat{H}$ is {\em right-seg\-ment\-ed} if all segments are right-leaning, and there are no edges other than those mandated by this fact. In other words, each forward source has all possible forward edges, and each vertex which is not a forward source has no forward edges. Similarly, we say that a path-separable signed graph $\widehat{H}$ is {\em left-segmented} if all segments are left-leaning, and there are no edges other than those mandated by this fact. In other words, each backward source has all possible backward edges, and each vertex which is not a backward source has no backward edges. Finally, $\widehat{H}$ is {left-right-segmented} if there is a unique segment $v_iv_{i+1} \ldots v_{i+2j+1}$ that is both left-leaning and right-leaning, all segments preceding it are left-leaning, all segments following it are right-leaning, and moreover there are {\em additional} bicoloured edges $v_{i-e}v_{i+2j+o}$ for all even $e \geq 2$ and all odd $o \geq 3$, but no other edges. In other words, vertices $v_1, v_2, \ldots, v_{i+2j+1}$ induce a left-segmented graph, vertices $v_i, v_{i+1}, \ldots, v_n$ induce a right-segmented graph, and in addition to the edges this requires there are all the edges joining $v_{i-e}$ from $v_1, \ldots, v_{i-1}$ to $v_{i+o}$ from $v_{i+2j+2}, \ldots, v_n$, with even $e$ and odd $o$. A {\em segmented graph} is a path-separable signed graph that is right-segmented or left-segmented or left-right-segmented. Note that a blue path without any bicoloured edges is trivially segmented, having no segments at all.

For example, in \Cref{fig:pathsegment} there are three segments, the left-leaning segment $v_5v_6v_7v_8v_9v_{10}$, the left- and right-leaning segment $v_{12}v_{13}v_{14}v_{15}$, and the right-leaning segment $v_{15}v_{16}v_{17}v_{18}v_{19}v_{20}$. Thus this is a left-right-segmented signed graph.

\begin{figure}[t]
\centering
\includegraphics[width=\textwidth]{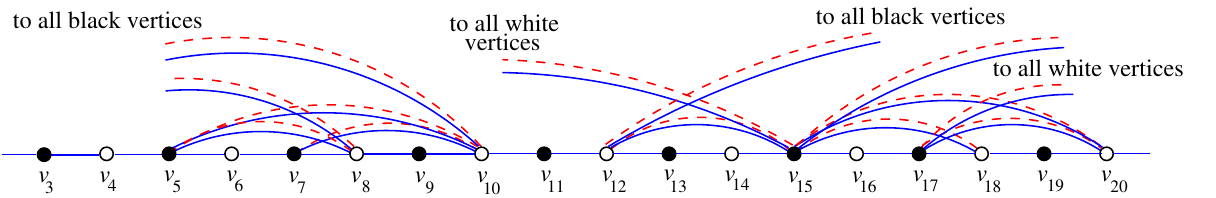}
\caption{An example of a left-right-segmented signed graph. The additional bicoloured edges from all white vertices before $v_{12}$ to all black vertices after $v_{15}$ are not shown.}
\label{fig:pathsegment}
\end{figure}

We are now ready to formulate the first of the two main results of this paper.

\begin{theorem}\label{hlavna}
Let $\widehat{H}$ be a path-separable signed graph. Then
\lhom{\widehat H} is polynomial-time solvable if
$\widehat{H}$ is a segmented signed graph. Otherwise, the problem is NP-complete.
\end{theorem}

We now focus on a proof of this theorem. For this case, we can use \Cref{mak} because, as we noted above, a path separable graph is semi-balanced.

\subsection{The NP-complete cases}

\begin{theorem}\label{thm:path-npc}
Let $\widehat{H}$ be a path-separable signed graph which is not segmented. Then \lhom{\widehat H} is NP-complete.
\end{theorem}

\begin{proof}
Let $\widehat{H}$ be a path-separable signed graph for which the list homomorphism problem is not NP-complete, with the unicoloured edges all blue and forming the spanning path $P = v_1, \ldots, v_n$. We will show that $\widehat{H}$ is a segmented signed graph. It follows from the discussion above that if there are bicoloured edges, then there are segments, but the problem is that there could be possibly some other bicoloured edges.

\paragraph{A crucial observation.} Consider two consecutive segments, a segment $S$, ending with the block $v_i$, $v_{i+1}$, $v_{i+2}$, $v_{i+3}$, and a segment $S'$ beginning with the block $v_j$, $v_{j+1}$, $v_{j+2}$, $v_{j+3}$, where $i+3 \leq j$. Note that there can be no bicoloured edge joining two vertices from the set $\{v_{i+1}, \ldots, v_{j+2}\}$, because the segments $S, S'$ are maximal and consecutive. If there is any edge joining two vertices of that set, there would have to be one forming a $4$-cycle with the unicoloured edges, since the underlying graph has no induced cycles longer than $4$.

We emphasize that this crucial observation will be repeatedly used in the arguments in the following paragraphs, usually without specifically mentioning it.

\paragraph{A claim.} We claim that either $v_i$ (the last forward source of the segment $S$) has forward edges to all $v_{i+3}, v_{i+5}, v_{i+7}, \ldots, v_{s}$ for some $s > j+1$, or symmetrically, $v_{j+3}$ (the first backward source of the next segment $S'$)  has backward edges to all $v_j, v_{j-2}, v_{j-4}, \ldots, v_{t}$ for some $t < i+2$. In the former case we say that $S$ {\em precedes} $S'$, in the latter case we say that $S'$ {\em precedes} $S$.

We now prove the claim by showing either $S$ precedes $S'$ or vice versa.

\paragraph{Case 1.} Suppose first that $i$ and $j$ have the same parity (are both even or both odd). There must be other bicoloured edges, otherwise there is a signed graph of type a) from the family $\cal F$ in \Cref{fig:forbgraphsirref} induced on the vertices $v_i$, $v_{i+3}$, $v_{i+4}, \ldots, v_j$, $v_{j+3}$, and hence a chain in $\widehat{H}$. Therefore, using the above crucial observation, there must be extra edges incident with $v_i$ or $v_{j+3}$. Moreover, as long as there is no edge $v_iv_{j+3}$, there would always be an induced subgraph of type a) from the family $\cal F$.
On the other hand, if $v_iv_{j+3}$ is an edge, then there is a chain with $U = v_i, v_{i+1}, v_{i+2}, v_{i+3}, v_i, v_{j+3}$ and $D = v_i, v_{j+3}, v_{j+2}, v_{j+1}, v_{j+2}, v_{j+3}$, unless $v_{i+2}v_{j+3}$ or $v_iv_{j+1}$ is an edge as there is no edge $v_{i+3}v_{j+2}$ by our crucial observation above. Note that both $v_{i+2}v_{j+3}$ and $v_iv_{j+1}$ cannot be edges, because of the chain $U = v_i, v_{i+1}, v_{i+2}, v_{j+3}$, $D = v_i, v_{j+1}, v_{j+2}, v_{j+3}$. Assume that $v_iv_{j+1}$ is an edge. Now we can repeat the argument: there would be a chain with $U = v_i, v_{i+1}, v_{i+2}, v_{i+3}, v_i, v_{j+1}$ and $D = v_i, v_{j+1}, v_j, v_{j-1}, v_j, v_{j+1}$, unless $v_iv_{j-1}$ is an edge. In this case, we don't need to consider $v_{i+2}v_{j+1}$, since both lie in $\{v_{i+1}, \ldots, v_{j+2}\}$. We can continue this way until this argument implies the already existing edge $v_iv_{i+3}$, and conclude that $v_i$ is adjacent to all $v_{j+3}, v_{j+1}, v_{j-1}, \ldots, v_{i+3}$, which proves the claim with $s=j+3$. If $v_{i+2}v_{j+3}$ is an edge, we conclude symmetrically that the claim holds for $t=i$.

\paragraph{Case 2.} Now assume that the parity of $i$ and $j$ is different. This happens, for instance, when $i+3 = j$: in this case, we would have a $4$-cycle pair unless one of  $v_iv_{j+2}, v_{i+1}v_{j+3}$ is an edge. Both cannot be edges, as there would be an alternating $4$-cycle. This verifies the claim when $i+3=j$. Otherwise, there again is a signed graph of type a) from the family $\cal F$ in \Cref{fig:forbgraphsirref}, induced on the vertices $v_i, v_{i+3}, v_{i+4}, \ldots, v_j, v_{j+3}$, so there must be extra edges incident with $v_i$ or $v_{j+3}$. Moreover, there would always remain such an induced subgraph unless there is a vertex $v_p$ with $i+3 \leq p \leq j$ that is adjacent to both $v_i$ and $v_{j+3}$. Thus let $v_p$ be such a vertex. 

We first show that $p$ can be chosen to be $j$ or $i+3$, i.e., that $v_j$ is adjacent to $v_i$, or $v_{j+3}$ is adjacent to $v_{i+3}$. Indeed, if $v_iv_j, v_{i+3}v_{j+3}$ are not edges, then $v_iv_{j+2}$ must also not be an edge (else we obtain a cycle of length greater than $4$), and we have the chain $$U = v_i, v_{i+1}, v_{i+2}, v_{i+3}, v_i, v_p, v_{j+3},\quad D = v_i, v_p, v_{j+3}, v_{j+2}, v_{j+1}, v_{j+2}, v_{j+3}.$$ Recall that we are still using the crucial observation that there is no bicoloured edge joining two vertices from the set $\{v_{i+1}, \ldots, v_{j+2}\}$.

Consider now the case that $p=j$ (or symmetrically $p=i+3$). Since $v_iv_p$ is an edge and there are no induced cycles of length greater than $4$, the vertex $v_i$ must also be adjacent to $v_{j-2}, v_{j-4}, \ldots, v_{i+5}$. Moreover, using again the crucial observation, $v_i$ is also adjacent to $v_{j+2}$, as otherwise we would have the chain $U = v_i, v_{i+1}, v_{i+2}, v_{i+3}, v_i, v_j$ and $D = v_i, v_j, v_{j+1}, v_{j+2}, v_{j+1}, v_j$. Thus the claim holds, with $s=j+2$. In the case $p=i+3$, we obtain a symmetric situation, proving the claim with $t=i+1$.

We conclude that for any two consecutive segments, exactly one precedes the other.

\paragraph{Auxiliary segments.} For technical reasons, we also introduce two {\em auxiliary} segments, calling all other segments {\em normal}. If the first normal segment $S$ of $\widehat{H}$ starts at $v_i$ with $i > 2$, we introduce the {\em left end-segment} to consist of the vertices $v_1, v_2, \ldots, v_i$. We say that the left end-segment {\em precedes}  $S$ if there is no edge $v_{i-2}v_{i+3}$, and we say that $S$ precedes the left end-segment if $v_{i-2}v_{i+3}$ is an edge. Similarly, if the last normal segment $S'$ ends at $v_k$ with $k < n-1$, the {\em right end-segment} consists of the vertices $v_k, v_{k+1}, \ldots, v_n$. The right end-segment {\em precedes}  $S'$ if $v_{k-3}v_{k+2}$ is not an edge, and $S'$ precedes the right end-segment if $v_{k-3}v_{k+2}$ is an edge. Then it is still true that for any two consecutive segments, one precedes the other.

\paragraph{A special situation.} Suppose that we have the special situation where \emph{each segment (including the end-segments) precedes the next segment}. Consider again the last normal segment $S'$, ending in block $v_{k-3}, v_{k-2}, v_{k-1}, v_k$. We first note that since there are no induced cycles of length greater than $4$, and no blocks after the block $v_{k-3}, v_{k-2}, v_{k-1}, v_k$, there cannot be any forward edges from $v_{k-1}, v_k, v_{k+1}, \ldots$. By the same argument and the absence of alternating $4$-cycles, there are no forward edges from $v_{k-2}$ either. 

Our special assumption implies that $v_{k-3}v_{k+2}$ is an edge. Then $v_{k-3}$ has also an edge to $v_{k+4}$, otherwise we have a signed graph of type b) from family $\cal F$ in \Cref{fig:forbgraphsirref}, induced on the vertices $v_{k-1}, v_{k-2}, v_{k-3}, v_{k+2}, v_{k+3}, v_{k+4}$. It is easy to check that the subgraph is induced because otherwise there would be either another block, or an alternating $4$-cycle. Then we argue similarly that $v_{k-3}$ has also an edge to $v_{k+6}$, and so on, concluding by induction that the last forward source $v_{k-3}$ has all possible forward edges, and that no vertex after $v_{k-3}$ has any forward edges. Because of the absence of alternating $4$-cycles, also the vertex $v_{k-4}$ has no forward edges, so if $v_{k-5}$ is a forward source in $S'$, we can use the same arguments to conclude it has all forward edges. Thus each forward source $v_{k-o}$ of $S'$ (with odd $o > 3$) has all possible forward edges. We conclude the last segment $S'$ is right-leaning, and there are no other forward edges (starting in its vertices or later) than those mandated by this fact.

We proceed by induction from the last segment to the first segment to show that in this special situation all segments are right-leaning and there are no other forward edges at all. The proof is analogous to the preceding paragraph. Consider for instance a segment $S$ ending in block $v_i, v_{i+1}, v_{i+2}, v_{i+3}$: since it precedes the next block, its last forward source, $v_i$ has all possible forward edges until $v_s$ where $s > j+1$ and the next segment begins with $v_j$. Now the arguments can be repeated, starting with avoiding a signed graph of type b) from family $\cal F$ in \Cref{fig:forbgraphsirref} induced on the vertices $v_{i+2}, v_{i+1}, v_i, v_s, v_{s+1}, v_{s+2}$. Finally, the vertices in the left end-segment cannot have any forward edges, as the absence of other blocks and of induced cycles of length greater than $4$ implies there would have to be an edge $v_{f-5}v_f$ where $v_f$ is the first backward source, contrary to the assumption that the left end-segment precedes the first normal segment.

Thus we have proved that if each segment precedes the next segment, then $\widehat{H}$ is a right-segmented graph. By symmetric arguments, we obtain the case of left-segmented signed graphs by assuming that each segment precedes the previous segment. It remains to consider the cases where some segment precedes, or is preceded by, both its left and right neighbours. \emph{Our goal is to prove that in that case, \lhom{\widehat H} is NP-complete if $\widehat H$ is not left-right-segmented.}

\paragraph{The general situation.} It turns out that the case where two segments $S_1$ and $S_3$ both precede the intermediate segment $S_2$ is impossible. Suppose $S_2$ has vertices $v_a, v_{a+1}, \ldots, v_b$; in particular, this implies that $v_av_b$ is an edge. Since each segment before $S_2$ precedes the next segment, the previous arguments apply to the portion of the vertices before $S_2$, and in particular the vertex $v_{a-2}$ is not a forward source, hence has no forward edges; therefore $v_{a-2}$ is not adjacent to $v_b$. By a symmetric argument, there is no edge $v_av_{b+2}$. There is also no edge $v_{a-1}v_{b+1}$ because it would form an alternating $4$-cycle with $v_av_b$. Therefore, $v_{a-2}, v_{a-1}, v_a, v_b, v_{b+1}, v_{b+2}$ induce a signed graph of type b) from family $\cal F$ in \Cref{fig:forbgraphsirref}, a contradiction.

We conclude that if $S_1$ precedes the next segment $S_2$, then $S_2$ must precede the following segment $S_3$, and so on, and similarly if $S_3$ precedes the previous segment~$S_2$.

Hence it remains only to consider the situation where a unique segment $S_2$, with vertices $v_a, v_{a+1}, \ldots, v_b$ precedes both its left neighbour $S_1$ and its right neighbour $S_3$ and to the left of $S_1$ each segment precedes its left neighbour, and to the right of $S_3$ each segment precedes its right neighbour. This implies that all segments before $S_2$ are left-leaning, all segments after $S_2$ are right-leaning, while $S_2$ is both left-leaning and right-leaning.

To prove that in this situation the signed graph $\widehat{H}$ is left-right-segmented, we show that all edges $v_{a-e}v_{a+o}$ are present, with $e$ even and $o$ odd. This is obvious when $e=0$ and $o \leq b-a$, by the observations following the definition of a segment. We also have the edges $v_av_{b+2}, v_{a-2}v_b$ since the segment $S_2$ is both left-leaning and right-leaning. Then we must have the edge $v_{a-2}v_{b+2}$ else there would be the chain $U = v_a, v_{a-1}, v_{a-2}, v_b, D = v_a, v_{b+2}, v_{b+1}, v_b$. We have the edge $v_av_{b+4}$ since $v_a$ is a forward source, and thus we must also have the edge $v_{a-2}v_{b+4}$; otherwise, there is the chain $U = v_a, v_{a-1}, v_{a-2}, v_{b+2}, D = v_a, v_{b+4}, v_{b+3}, v_{b+2}$. Continuing this way by induction on $e+o$ we conclude that all edges $v_{a-e}v_{a+o}$, with $e$ even and $o$ odd, must be present. This completes the proof of NP-completeness for any path-separable signed graph that is not segmented.
\end{proof}

\subsection{The polynomial cases}\label{sec:polynomial}

To show that \lhom{\widehat H} is polynomial when $\widehat{H}$ is a segmented signed graph, we show that $\widehat{H}$ has a special bipartite min ordering.

\begin{theorem}\label{thm:path-poly}
Let $\widehat{H}$ be a path-separable signed graph. If
$\widehat{H}$ is a segmented signed graph, then $\widehat{H}$ has a special bipartite min ordering and thus, 
\lhom{\widehat H} is polynomial-time solvable.
\end{theorem}

\begin{proof}
We distinguish two cases: $\widehat H$ being right-segmented or left-right-seg\-men\-ted.

\paragraph{Case 1.} We now describe \emph{a special bipartite min ordering of the vertices for the case of a right-segmented signed graph}, with the unicoloured path $P$.  Consider two white vertices $u$ and $v$ that are forward sources such that $u$ precedes $v$ on $P$. Then all forward neighbours of $v$ are also forward neighbours of $u$, and all backward neighbours of $u$ are also backward neighbours of $v$.  White vertices $z$ that are not forward sources have edges from all black forward sources $w$ that precede $z$ on $P$.  We now construct a bipartite min ordering $<$: order the white vertices that are forward sources in the forward order, then order the remaining white vertices in the backward order.  The same ordering is applied on black vertices. It now follows from our observations above that this is a special min ordering. For left-segmented graphs, the ordering is similar.

\paragraph{Case 2.} We now describe {\em a special min ordering $<$ for the case of a left-right-seg\-men\-ted signed graph $\widehat{H}$}; we consider its vertices in the order of the unicoloured path $P$. We may assume the left-right-leaning segment begins with a white vertex $a$ and ends with a black vertex $b$. We denote by $a'$ the (black) successor of $a$ on $P$, and by $b'$ the (white) successor of $a'$ on $P$. We also denote by $L$ the set of backward sources and by $R$ the set of forward sources of $\widehat{H}$, and denote by $U$ the portion of $P$ from its first vertex to $a'$, and by $V$ the remaining portion of $P$, from $b'$ to its last vertex. Then the min ordering $<$ we construct has white vertices ordered as follows: first the vertices in $U \cap L$ listed in backward order on $P$, then the vertices of $U \setminus L$ listed in forward order on $P$, followed by the vertices in $V \cap R$ in forward order, and then the vertices of $V \setminus R$ in backward order. Similarly, the black vertices are ordered as follows: first the vertices of $V \cap R$ in forward order, then vertices of $V \setminus R$ in backward order, then vertices of $U \cap L$ in backward order, and finally the vertices of $U \setminus L$ in forward order.

To check that $<$ is a min ordering we consider where could a violation of the min ordering property lie. A violation would consist of white vertices $x < y$ and black vertices $s < t$ such that $xt, ys$ are edges of $\widehat{H}$ but $xs$ is not.

We observe that all white vertices in $U$ join all black vertices in $V$, and no black vertex in $U$ joins a white vertex in $V$ with the exception of $a'$ joining $b'$.

\begin{enumerate}
	\item \emph{$x$ lies in $U \cap L$:} In this case, $x$ is adjacent to all black vertices in $V$, all black vertices before $x$ on $P$ and to its immediate successor and predecessor on $P$. As $x$ belongs to $L$, its predecessor and its successor belong to $U \backslash L$ and are its only unicoloured neighbours.  Therefore the non-neighbours of $x$ are in $U \backslash L$, and ordered in $<$ later than its neighbours. A violation cannot occur.  Moreover, we observe that its bicoloured neighbours all precede its unicoloured neighbours in $<$ and hence we also verify the special property of a min ordering for $x$.

	\item \emph{$x$ lies in $U \backslash L$:} This situation implies $s \in U$ and since $s < t$, $t \in U$ also.  If $s$ follows $x$ on $P$, then $s \in U \backslash L$ giving $t \in U \backslash L$ which means $xt$ cannot be an edge.  On the other hand, if $s$ precedes $x$ on $P$, the edge $ys$ implies $y \in U$ and the ordering $x < y$ implies $y \in U \backslash L$ giving $ys$ cannot be an edge.  It is easy to check that the special property holds in this case also. 

	\item \emph{$x \in V \backslash R$:}  Since $x < y$, we must have $y \in V \backslash R$ as well and $y$ precedes $x$ on $P$.  The only black vertices adjacent to $y$ with a bicoloured edge must lie to the left of $y$ on $P$ and therefore to the left of $x$. Such vertices join $x$ with a bicoloured edge as well.  Thus $s \not\in R$ and must be adjacent to $y$ by a unicoloured edge. Since $s < t$ and $t \in V$ (no black vertex in $U$ can join $x$), we have $t$ precedes $s$ on $P$ and $t \not\in L$.  Now $tx$ cannot be an edge.  The special property is straightforward to verify.

	\item \emph{$x \in V \cap R$:} Since $x$ is a forward source $s$ must precede $x$ on $P$. If $s \in V$, then $s \in V \backslash R$ as $sx$ is not an edge.  Since $xt$ is an edge, $t \in V$ and $s < t$ implies $t \in V \backslash R$ contradicting $xt$ is an edge.  If $s \in U$, then $t \in U$ contradicting the existence of at least one of $ys$ or $xt$.
\end{enumerate}

It is easy to check that the special property holds in all cases, and that it also holds for all black vertices.
Thus our ordering is a special min ordering and we can use the result from~\cite{ks,latin} which asserts that for semi-balanced signed graphs, the existence of a special min ordering ensures the existence of a polynomial-time algorithm. 
\end{proof}

\section{Cycle-separable signed graphs}

As an application of \Cref{hlavna}, we now consider irreflexive signed graphs in which the unicoloured edges form a spanning cycle $C$. Recall that we say that an irreflexive signed graph $\widehat{H}$ is {\em cycle-separable} if the unicoloured edges of $\widehat{H}$ form a spanning cycle in the underlying graph $H$. In other words, we have a spanning cycle $C$ whose edges are all unicoloured, and all the other edges of $\widehat{H}$ are bicoloured.

In contrast to the path-separable signed graphs, we cannot assume the edges of $C$ are all blue, and hence our signed graphs in this section are not semi-balanced, and we cannot use \Cref{mak}.
However, we note that the spanning cycle $C = v_0v_1 \ldots v_{n-1}$ is again unique.

\paragraph{The polynomial cases.} We first introduce three cycle-separable signed graphs (actually one is a family of graphs) for which the list homomorphism problem will turn out to be polynomial-time solvable. The signed graph $\widehat{H_0}$ is the $4$-cycle with all edges unicoloured blue. The signed graph $\widehat{H_1}$ consists of a blue path $b=t_0, t_1, t_2, t_3=w$, a blue-red-blue path $b, s_1, s_2, w$, together with a bicoloured edge $bw$. The signed graph $\widehat{H_\ell}$ consists of a blue path $b, s_1, s_2, w$, a blue path $b=t_0, t_1, t_2, \dots, t_{\ell}=w$ (with $\ell \geq 3$ odd), and all bicoloured edges $t_it_j$ with even $i$ and odd $j, j > i+1$. (Note that this includes the edge $bw$.) These three cycle-separable signed graphs $H_k, k=0, 1, \ell$, are illustrated in \Cref{fig:segment23}. Note that if the subscript $k$ is greater than $0$, then it is odd. Moreover, both $H_1$ and $H_3$ have $6$ vertices and differ only in the colours of the unicoloured edges forming the spanning cycle $C$: $H_1$ has the cycle $C$ unbalanced, and $H_3$ has the cycle $C$ balanced.  

\begin{figure}
\centering
\includegraphics[width=0.95\textwidth]{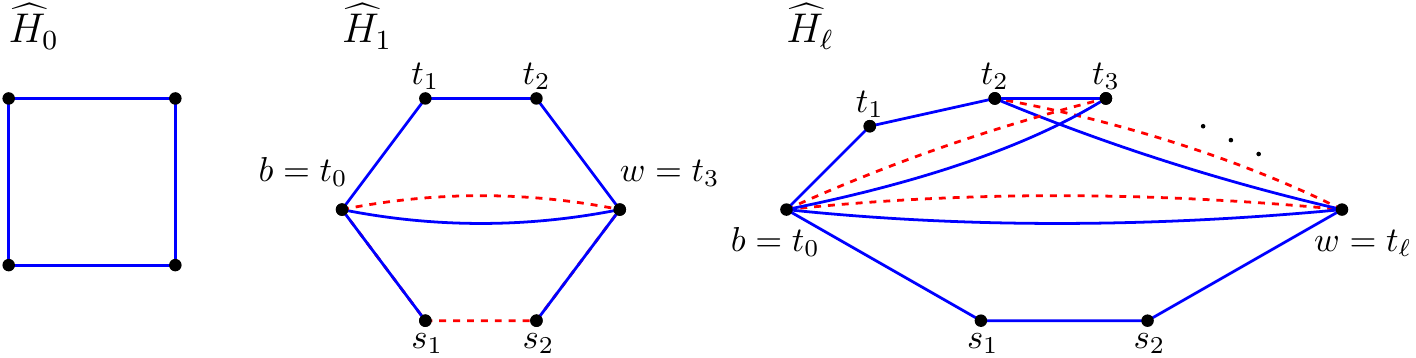}
\caption{The cycle-separable signed graphs $\widehat{H}_0$,  $\widehat{H}_1$, and $\widehat{H}_\ell$ with $\ell \geq 3$ odd.}
\label{fig:segment23}
\end{figure}

For the rest of this section, our goal is to prove the following theorem.

\begin{theorem} \label{thm:Hamilton}
Let $\widehat{H}$ be a cycle-separable signed graph.
Then \lhom{\widehat H} is polynomial-time solvable if $\widehat{H}$ is switching equivalent to $\widehat{H}_0$,  or to $\widehat{H}_1$, or to $\widehat{H}_\ell$ for some odd $\ell \geq 3$.
Otherwise, the problem is NP-complete.
\end{theorem}

\subsection{The NP-complete cases}

We prove the following theorem by contrapositive, assuming \lhom{\widehat H} is not NP-complete. We will provide a series of observations excluding more obvious cases which will leave us with graphs from \Cref{fig:segment23} and with one additional family of cycle-separable graphs. We will show an NP-completeness reduction for the latter.

\begin{theorem} \label{thm:cycle-npc}
Let $\widehat H$ be a cycle-separable signed graph.
If $\widehat H$ is not switching equivalent to any of the graphs from \Cref{fig:segment23}, then \lhom{\widehat H} is NP-complete.
\end{theorem}

\begin{proof}
Suppose $\widehat{H}$ is a cycle-separable signed graph for which \lhom{\widehat H} is not NP-complete. As $\widehat{H}$ is irreflexive, it must be bipartite.  If $\widehat{H}$ has no bicoloured edges, then it must be a balanced even cycle of length $4$, i.e., $\widehat{H}$ is switching equivalent to $\widehat{H_0}$.

Now suppose that $\widehat{H}$ has at least $6$ vertices and the spanning cycle $C$ with cyclically ordered vertices $v_0, v_1, \dots, v_{n-1}, v_0$. Without loss of generality all edges are blue with the possible exception of $v_{0}v_1$.  Consider $\widehat{H}-v_0$.  
This is switching equivalent to a segmented signed graph with spanning path $v_1, v_2, \dots, v_{n-1}$ (where $n$ is even), and thus it has the structure described above.

By symmetry, there is a right-leaning segment.  Let $v_i$ be the first vertex of the first right-leaning segment.  Then $v_{i+1}$ has degree $2$ in $\widehat{H}-v_0$ and degree $2$ in $\widehat{H}$ unless $v_0 v_{i+1}$ is a bicoloured edge.  If $v_0 v_{i+1}$ is an edge, then $i+1$ is odd and the forward source $v_i$ sends a bicoloured edge to each $v_o$ for $o$ odd with $i+3 \leq o \leq n-1$.  Consequently, $v_0 v_{i+1} v_i v_{n-1} v_0$ is an alternating $4$-cycle contrary to our assumption that  $\textsc{List-S-Hom}(\widehat{H})$ is not NP-complete. We conclude $v_{i+1}$ has degree 2 in $\widehat{H}$.

Rename the vertices of the underlying graph $H$ so that $v_0$ is a vertex of degree two. The signed graph $\widehat{H}-v_0$ is path-separable. We are assuming that the list homomorphism problem for $\widehat{H}-v_0$ is not NP-complete, so \Cref{hlavna} implies that $\widehat{H}-v_0$ is switching equivalent to a segmented signed graph with spanning path $P = v_1, v_2, \dots, v_{n-1}$.
In particular, we may switch so the spanning path $v_1, v_2, \dots, v_{n-1}$ of unicoloured edges is all blue, the edge $v_0v_{n-1}$ is blue, and the edge $v_0v_1$ may be red of blue (depending on the sign $C$). 

By symmetry, we may assume $v_2$ is adjacent to $v_{n-1}$ (recall $n$ is even), otherwise $\widehat{H}$ contains an induced cycle of length greater than four. Thus we have a $4$-cycle $v_0, v_1, v_2, v_{n-1}, v_0$. 

Assume first that $v_2, v_3, v_4, v_5, v_2$ is also a 4-cycle. Then we must have that $v_4v_{n-1}$ is also an edge, otherwise $\widehat{H}$ would have 4-cycle pair. If $v_1v_4$ is an edge, then we have an alternating $4$-cycle.  If $v_6 v_{n-1}$ is not an edge, then the path $v_6, v_5, v_4, v_{n-1}, v_0, v_1$ is a case $(b)$ of family ${\mathcal F}$ in \Cref{fig:forbgraphsirref}. Similarly, $v_1 v_6$ is not an edge and $v_2 v_7$ is.  By repeating this argument, we conclude that $v_8v_{n-1}, \ldots , v_{n-4}v_{n-1}$ and $v_2v_9, \ldots, v_2v_{n-3}$ must also be edges. Thus using our descriptions of the polynomial path-separable cases, we conclude that $\widehat{H}- v_0 - v_1 = v_2, v_3, \ldots, v_{n-1}$ is just one segment. If $v_{n-1}, v_{n-2}, v_{n-3}, v_{n-4}$ is a 4-cycle, the argument is similar.

If neither $v_2, v_3, v_4, v_5$ nor $v_{n-1}, v_{n-2}, v_{n-3}, v_{n-4}$ is a 4-cycle, then from \Cref{hlavna} we conclude that $\widehat{H}-v_0$ is left-right-segmented, with a left-right-leaning segment $S$ not at the end of $P$. 
This is easy to dismiss, because there would be a $P_5$ (case (b) of family $\mathcal{F}$) involving the segment $S$ and the vertex $v_0$. In conclusion $\widehat{H}- v_0 - v_1$ is just one segment.  We label the vertices of $\widehat{H}$ as in \Cref{fig:segment23}, namely, the segment $v_2, \dots, v_{n-1}$ is $b=t_0, t_1, \dots, t_{\ell}=w$ and the path $v_2, v_1, v_0, v_{n-1}$ is $b, s_1, s_2, w$.

If the cycle $C$ is balanced, then $\widehat{H}$ is switching equivalent to some $\widehat{H_\ell}$ with $\ell \geq 3$. If the cycle $C$ is unbalanced and $n=6$, then $\widehat{H}$ is switching equivalent to $\widehat{H_1}$. In fact, $\widehat{H}$ has only the edge $s_1s_2$ red. We will show later (in \Cref{thm:cycle-poly}) that both these cases are polynomial-time solvable.

\paragraph{The main case.} Therefore, assume that $\widehat{H}$ has $n > 6$ and its spanning cycle $C$ is unbalanced. Without loss of generality, assume the path $b, t_1, \dots, t_{\ell-1}, w$ is blue and the path $b, s_1, s_2, w$ is red.  The vertices of the segment $b,s_1,s_2,w$ are called the \emph{$s$-vertices} and the vertices of the segment $b,t_1,\ldots,t_{\ell-1},w$ are called the \emph{$t$-vertices}.  

We now prove that in this case \lhom{\widehat H} is NP-complete. Unfortunately, we cannot use a chain (the graph is not semi-balanced so \Cref{mak} does not apply). We will instead reduce from one of the NP-complete cases of Boolean satisfiability dichotomy theorem of Schaefer~\cite{schaefer1978complexity}. 

\paragraph{The problem.} An instance of the problem is a set of Boolean variables $V$ and a set of quadruples $R$ over these variables. The problem asks if there is an assignment of $0, 1$ to the variables so that for every quadruple $(a',b',c',d') \in R$, the Boolean expression $(a' = b' = c' = d') \vee (a' \neq c')$ is satisfied.

Schaeffer~\cite{schaefer1978complexity} proved that a Boolean constraint satisfaction problem is NP-complete except for the well known polynomial cases of 2-SAT, Horn clauses, co-Horn clauses, linear equations modulo two, or when the only satisfying assignments are the all true or the all false assignments. To see that our problem is not expressible as 2-SAT, consider the following three satisfying assignments for $(a',b',c',d')$:  $(1,1,1,1), (1,0,0,0), (0,0,1,0)$. 

It is well known, see e.g.~\cite{sudan} Lemma 4.9, that any problem expressible as 2-SAT has the property that the majority function on three satisfying assignments must also be a satisfying assignment. However, for our three assignments the majority function yields the assignment $(1,0,1,0)$ which is not satisfying. 
Similarly, our problem is not expressible as Horn clauses (respectively co-Horn clauses) because the minimum (respectively maximum) function on the two satisfying assignments $(0,1,1,0), (1,1,0,0)$ is not a satisfying assignment, cf.\ Lemma 4.8 in~\cite{sudan}. 
Finally, our problem is not expressible by linear equations modulo two because the sum modulo two of the three satisfying assignments $(1,1,1,1), (1,1,0,1), (0,1,1,1)$ results in the assignment $(0,1,0,1)$ which is not satisfying, cf.\ Lemma 4.10 in~\cite{sudan}. Thus our problem is one of the NP-complete cases. 

\paragraph{Gadgets.} Consider an instance $R$ (over $V$) of our satisfiability problem. We shall now construct a signed graph $\widehat{G}$ with lists such that $\widehat{G}$ admits a list homomorphism to $\widehat{H}$ if and only if there is a satisfying assignment for the set of quadruples $R$.

For each quadruple $(a',b',c',d')$, we shall construct a copy of the gadget $Q(a',b',c',d')$ (with lists) as in \Cref{fig:quad}. 
Observe that the images of $a', b', c'$ and $d'$ are all fixed.  The remaining vertices, which we call \emph{inner vertices}, must all map to the $t$-vertices or must all map to the $s$-vertices.

A variable, say $r$, can appear in multiple quadruples.  In this case, there will be a vertex corresponding to $r$ for each quadruple.  If $r$ appears multiple times in the first or second coordinate, then we add a vertex $x_r$ to $\widehat{G}$ and a blue edge from $x_r$ to each occurrence of $r$ (in the first two coordinates).  The vertex $x_r$ has the list $\{ t_1 \}$.  Similarly, if $r$ appears in the last two coordinates (corresponding to $c'$ and $d'$) of multiple quadruples, then a vertex $y_r$ with list $\{ t_{\ell-1} \}$ is added to $\widehat{G}$ together with blue edges joining $y_r$ to each occurrence of $r$ in the last two coordinates.  Finally if $r$ occurs in the first or second coordinate in one quadrangle, say as $r'$,  and in the third or fourth coordinate of another quadrangle, as $r''$, then a blue path $r', r_1, r_2, \dots, r_{\ell-1}, r''$ is added to $\widehat{G}$ with $L(r_i) = \{ t_i \}$ for each $i=1, 2, \dots, \ell-1$.  This path needs only to be added once for the variable $r$.  Observe at this point between any two occurrences of $r$ in $\widehat{G}$, there is a blue path whose image under any homomorphism to $\widehat{H}$ is uniquely determined by its lists.  Moreover, the image of each such path is a positive walk. Consequently, under any list homomorphism $\widehat{G} \to \widehat{H}$ either no occurrence of $r$ is switched or all occurrences of $r$ are switched.

\begin{figure}
\centering
\includegraphics[scale=1.1]{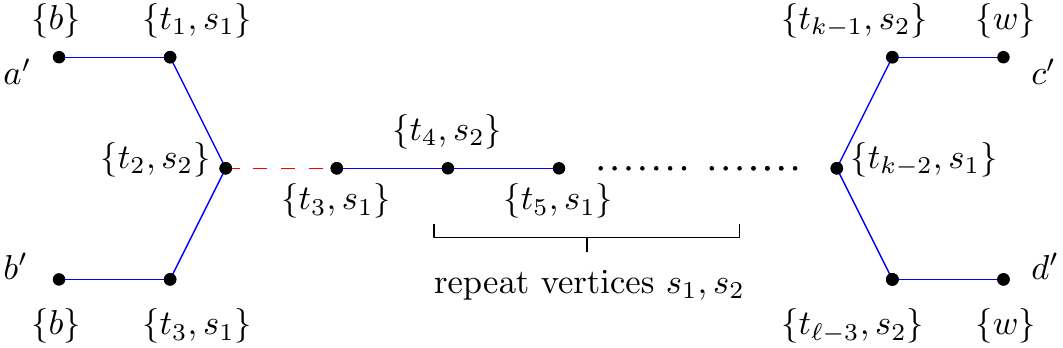}
\caption{A quadruple gadget $Q(a',b',c',d')$.}
\label{fig:quad}
\end{figure}

\paragraph{Proof of reduction.} We claim that there exists a satisfying assignment for $R$ if and only if there is a list homomorphism of $\widehat{G}$ to $\widehat{H}$.

Let $f\colon \widehat{G} \to \widehat{H}$ be a list homomorphism.  We define an assignment $\pi_f\colon V \to \{ 0, 1\}$ by setting the variable $r$ to $1$ if an occurrence of $r$ in $\widehat{G}$ is switched under the homomorphism $f$ and setting $r$ to $0$ otherwise. As observed above, under $f$, all occurrences of  $r$  are switched or no occurrence is switched. Thus the assignment $\pi_f$ is a well defined.

To complete the reduction we show $\pi_f$ is a satisfying truth assignment.  Consider a particular copy of the gadget $Q(a',b',c',d')$ and consider the quadruple $(\pi_f(a'), \pi_f(b'), \pi_f(c'), \pi_f(d'))$.  For vertices $u, v$ in $Q$, let $P(u,v)$ denote the path from $u$ to $v$ in the copy of $Q$.  Initially, $P(a',b')$ and $P(c',d')$ are both positive paths, while
$$P(a',c'), P(a',d'), P(b',c'), P(b',d')$$ are all negative.  Switching an end point of a path changes the sign of the path while switching an interior vertex of the path leaves its sign unchanged.
The crucial observation is that after we fix switchings at endpoints, the signs of the paths $P(a',c')$, $P(a',d')$, $P(b',c')$, $P(b',d')$ (let us call them \emph{main paths}) are invariant upon switching at some inner vertices of the gadget.

The last thing we need to argue is that certain switchings of the endpoints of the quadruple gadget are not possible and some of them are possible. We denote a particular switching as a quadruple $(s_{a'},s_{b'},s_{c'},s_{d'})$ with zeroes and ones with the meaning that one corresponds to not being switched and zero corresponds to being switched. 

The images of $a', b', c'$ and $d'$ under $f$ are uniquely determined by their lists.  The remaining vertices, which we call \emph{inner vertices}, must all map to the $t$-vertices or must all map to the $s$-vertices. That is, for a given quadruple gadget, its inner vertices must all choose either the first element described in its list or the second in every possible list homomorphism.  We consider the two cases.

\begin{itemize}
  \item \emph{The internal vertices map to the $t$-vertices.}  In this case $P(a',c')$ maps to $b, t_1, t_2, \dots, t_{\ell-1}, w$.  This path is positive in $\widehat{H}$ while $P(a',c')$ is negative in $Q$.  Thus exactly one of $a'$ or $c'$ must be switched under $f$.  That is, $\pi_f(a') \neq \pi_f(c')$ and $(\pi_f(a'), \pi_f(b'), \pi_f(c'), \pi_f(d'))$ is a satisfying truth assignment. The bicoloured edges $bt_3$ and $t_{\ell-3}w$ ensure $b'$ and $d'$ can be switched or not independently of $a', c'$ and each other.
  \item \emph{The internal vertices map to the $s$-vertices.} In this case all of $P(a',c')$, $P(a',d')$, $P(b',c'),$ and  $P(b',d')$ map to $b, s_1, s_2, s_1, \dots, s_2, w$.  As the four paths in $Q$ and the image (walks) in $\widehat{H}$ are all negative, it follows that either all of $\{ a', b', c', d' \}$ are switched or none is switched.  Thus $\pi_f(a') = \pi_f(b') = \pi_f(c') = \pi_f(d')$ and again we have a satisfying truth assignment for the quadruple. 
\end{itemize}

Conversely, assume we have a satisfying truth assignment, say $\pi\colon V \to \{ 0, 1 \}$.  For each variable $r$, switch all occurrences of $r$ if and only if $\pi(r) = 1$.  As observed above, all paths between occurrences of $r$ are (still) positive and admit a list homomorphism to $\widehat{H}$.  Consider a particular quadruple $Q(a',b',c',d')$. If $\pi(a') \neq \pi(c')$, then (after switching), the path $P(a',c')$ is positive; hence, we can switch internal vertices in $Q$ to make the path blue. We map it to $b, t_1, \dots, t_{\ell-1}, w$.  We map $P(a',b')$ to $b, t_1, t_2, t_3, b$.  Since the edge $bt_3$ is bicoloured, (after possibly switching at the neighbour of $b'$) this mapping is a list homomorphism.  A similar analysis works for the path $P(c',d')$.  Thus, $Q$ admits a list homomorphism of $\widehat{H}$. If $\pi(a')=\pi(b')=\pi(c')=\pi(d')$, then all or none of $a',b',c',d'$ are switched.  In this case, the internal vertices of $Q$ can be switched so all the edges are red.  Hence, $Q$ maps to the path $b, s_1, s_2, w$ which again is the desired list homomorphism.
\end{proof}

\subsection{The polynomial cases}

Next we show that the corresponding problem can be solved in polynomial time for all the remaining cycle-separable signed graphs, illustrated in \Cref{fig:segment23}.

\begin{theorem}\label{thm:cycle-poly}
\lhom{\widehat H} is polynomial-time solvable if $\widehat{H}$ is switching equivalent to $\widehat{H}_0$,  or to $\widehat{H}_1$, or to $\widehat{H}_\ell$ for some odd $\ell \geq 3$.
\end{theorem}

\begin{proof}
The proof is divided into three cases.

\medskip
\noindent\textit{Case 1:} If $\widehat{H}$ is switching equivalent to $\widehat{H}_0$, then \lhom{\widehat H} is polynomial-time solvable by \Cref{thm:bordeaux}.

\medskip
\noindent\textit{Case 2:} If $\widehat{H}$ is switching equivalent to $\widehat{H}_{\ell}$ with $\ell \geq 3, \ell$ odd, then \lhom{\widehat H} is polynomial-time solvable by \Cref{mak}. Specifically, we claim that $\widehat{H}_{\ell}$ has a special min ordering. To see this, remove the vertices $s_1$ and $s_2$, obtaining a path-segmented signed graph $\widehat{H'}$ consisting of just one segment. According to \Cref{sec:polynomial}, $\widehat{H'}$ has a special min ordering $<$ in which $b=t_0 < t_2 < t_4 < \dots$ and $w=t_{\ell} < t_{\ell-2} < t_{\ell-4} < \dots$. To obtain a special min ordering of $\widehat{H_{\ell}}$ we simply add the vertices $s_1, s_2$ at the end of $<$, i.e., we set $b=t_0 < t_2 < t_4 < \dots < s_2$ and $w=t_{\ell} < t_{\ell-2} < t_{\ell-4} < \dots < s_1$. 
The vertex $t_{2i}$, $i > 0$, has bicoloured edges to $t_{\ell}, t_{\ell-2}, \dots, t_{2i+3}$ and unicoloured edges to $t_{2i+1}, t_{2i-1}$.   Similarly for $t_{2i-1}$. Further, since $b$ and $w$ are adjacent with all vertices of the opposite colour, this is a special min ordering of $\widehat{H}_{\ell}$.

\medskip
\noindent\textit{Case 3:} It remains to handle the final case when $\widehat{H}$ is switching equivalent to $\widehat{H}_1$ In this case $\widehat H$ is not semi-balanced and hence a different technique is needed. Note however that $\widehat{H}_1$ does have a special min ordering, identical to that for $\widehat{H}_3$ --- $\widehat{H}_1$ and $\widehat{H_3}$ only differ in the colour of some unicoloured edges, which is irrelevant in the definition of special min ordering. The following technique, transforming the problem to solving a system of linear equations, is inspired by our proofs in the case of signed trees \cite{dm}, where more details can be found.

\paragraph{Preprocessing.} Let $\widehat{G}$ together with lists $L$ be an instance of \lhom{\widehat{H_1}}.  We may assume $G$ is connected and bipartite. We will call the vertices of parts of bipartition in  $\widehat{G}$ black and white as well. First, we try mapping the black vertices of $\widehat{G}$ to the black vertices of $\widehat{H_1}$. If that fails, we try mapping the white vertices of $\widehat{G}$ to the black vertices of $\widehat{H_1}$.  In the former case we remove all white (respectively black) vertices from the lists of the black (respectively white) vertices in $\widehat{G}$.  The latter case is analogous.  Thus we may assume we are mapping black vertices to black vertices and white vertices to white vertices.

Next we apply to the underlying graphs of $\widehat{G}, \widehat{H_1}$ the arc consistency procedure from~\cite{feder1998list} (see Algorithm 4 in \cite{hn1}); we also apply the same procedure to the graphs spanned by the bicoloured edges. If any list becomes empty then there is no list homomorphism of the underlying graphs and hence no list homomorphism of $\widehat{G}$ to $\widehat{H_1}$. Otherwise choosing the minimum of each list (with respect to the special min ordering) defines a list homomorphism $f$ of the underlying graphs, in which moreover bicoloured edges are taken to bicoloured edges. 

As a result, $f$ maps any bicoloured edge of $\widehat{G}$ to the edge $bw$ in $\widehat{H_1}$, Moreover, all vertices $v$ of $\widehat{G}$ which have $b$ or $w$ in their post-consistency lists (black vertices $v$ with $b \in L(v)$ and white vertices $v$ with $w \in L(v)$) must have $f(v)=b$ or $f(v)=w$ because they $b, w$ are the smallest black and white vertices respectively, in the special min ordering of $\widehat{H_1}$. We call all vertices with $f(v) = b$ or $f(v)=w$ {\em boundary vertices}, and the remaining vertices {\em interior vertices}. Thus interior vertices cannot map to $b, w$ by any homomorphism, as $b, w$ are not in their post-consistency lists.

The interior vertices of $\widehat{G}$ form a union of components. Consider such a component $K$. The subgraph of $\widehat{G}$ induced by $K$ is called a \emph{region}. For each region, either all its vertices map to $s_1, s_2$ or all map to $t_1, t_2$. 

Thus all the edges of the regions and all the edges joining interior vertices to boundary vertices are unicoloured.

\paragraph{Handling regions.} Let $K$ be a region. We will now decide whether $K$ will be mapped to $s_1, s_2$ or to $t_1, t_2$, and determine an appropriate switching of the boundary vertices to make this a homomorphism.

Under any list homomorphism to $\widehat{H_1}$ the image of $K$ is a single edge. In particular, $K$ must be balanced.  Hence we can switch vertices of $K$ so that the edges are all blue, and then identify the black vertices and identify the white vertices so that $K$ is now a single blue edge on vertices $k_1, k_2$. 
The region has boundary vertices consisting of $u_1, u_2, \ldots$ mapping to $b$ (adjacent to $k_1$) and $v_1, v_2, \ldots$ (adjacent to $k_2$) mapping to $w$. We need to switch the boundary vertices so that the  subgraph induced by region and its boundary vertices is a balanced subgraph that maps to $b, t_1, t_2, w$ or an anti-balanced subgraph that maps to $b, s_1, s_2, w$.

Note that under any suitable homomorphism, the input graph $\widehat{G}$ must be switched so that each walk from $v_i$ to $v_j$ whose internal vertices belong to $K$ is positive, and similarly for $u_i$ to $u_j$. To determine the switching we solve (in polynomial time) a system of linear equations modulo two. For each boundary vertex $v_i$, we will introduce a variable $x_i$ with value $1$ intended to mean the vertex is switched and $0$ meaning it is not switched. Similarly we introduce a variable $y_i$ for each vertex $u_i$.

If there is a positive walk from $v_i$ to $v_j$, we include the equation $x_i = x_j$, for all $i, j$. Similarly, if there is a negative walk, we include $x_i \neq x_j$. We proceed similarly for $u_1, u_2, \ldots$ and the variables $y_i$.

\paragraph{Linear equations.} Finally, we introduce linear equations to determine the signs of walks from $v_i$ to $u_j$. If $K$ can only map to $t$-vertices, then the walks between $v_i$ and $u_j$ must be positive.  We code the needed switching with $x_i = y_j$ if there is a positive walk from $v_i$ to $u_j$ in $\widehat{G}$ and $x_i \neq y_j$ if there is a negative walk.  Similarly, if $K$ can only map to $s$-vertices, we must switch to make all walks between $v_i$ and $u_j$ negative.  If $K$ has a choice, then we can use a variable $z_k$ that is $0$ if $K$ maps to $t$-vertices and $1$ if $K$ maps to $s$-vertices. For this situation, we include the equation $x_i = y_j + z_k$ if there is a positive walk from $v_i$ to $u_j$ and $x_i = y_j + z_k+1$ if there is a negative walk.  After repeating this construction for each region, we have our desired system of linear equations.

It is easy to see that if the system of linear equations has no solution, then there is no list homomorphism to $\widehat{H_1}$. Otherwise, a solution identifies a switching for all boundary vertices which allows each region to map to $\widehat{H_1}$.

Thus we have transformed the final case to a polynomial problem of a system of linear equations, completing the proof of Case 3.
\end{proof}

\section{Conclusion}

It seems difficult to give a full combinatorial classification of the complexity of list homomorphism problems for general signed graphs. For irreflexive signed graphs, which are in a sense the core of the problem, there is a conjectured classification in~\cite{ks}. Here we have obtained a full dichotomy classification in the special case of separable irreflexive signed graphs. The classification confirms the dichotomy conjecture of~\cite{ks} for this case, and also confirms that the only polynomial cases enjoy a special min ordering and the only NP-complete cases have chains or invertible pairs, as also conjectured in~\cite{ks}.

\section{Acknowledgements}

The first author received funding from the European Union's Horizon 2020 project H2020-MSCA-RISE-2018: Research and Innovation Staff Exchange and a partial support by the ANR (Agence nationale de la recherche) project GRALMECO (ANR-21-CE48-0004). The second author was supported by his NSERC Canada Discovery Grant. The fourth and fifth author were also partially supported by the fourth author's NSERC Canada Discovery Grant. The first and the fifth author were also supported by the Charles University Grant Agency project 370122. 

\bibliographystyle{plainurl}
\bibliography{bibliography}

\end{document}